\documentclass[english,13pt]{article}
\usepackage[T1]{fontenc}
\usepackage[latin1]{inputenc}
\usepackage{geometry}
\usepackage{float}
\usepackage{mathrsfs}
\usepackage{amsmath}
\usepackage{amssymb}
\usepackage{graphicx}
\usepackage{hyperref}
\hypersetup{
    colorlinks=true,
    linkcolor=blue,
    filecolor=magenta,      
    urlcolor=cyan,
    citecolor = red,
}

\makeatletter

\usepackage{algorithm,algpseudocode}

\usepackage{amsthm}

\usepackage{mathrsfs}

\usepackage{amsfonts}

\usepackage{epsfig}

\usepackage{bm}

\usepackage{mathrsfs}

\usepackage{enumerate}

\numberwithin{equation}{section}

\@ifundefined{definecolor}{\@ifundefined{definecolor}
 {\@ifundefined{definecolor}
 {\usepackage{color}}{}
}{}
}{}

\usepackage{subfig}\usepackage[all]{xy}

\newtheorem{theorem}{Theorem}[section]

\newcounter{hypA}
\newenvironment{hypA}{\refstepcounter{hypA}\begin{itemize}
  \item[({\bf A\arabic{hypA}})]}{\end{itemize}}
\newcounter{hypB}

\newcounter{hypD}

\newcounter{hypW}

\usepackage{babel}\date{}

\usepackage{babel}

\makeatother

\usepackage{babel}

\begin{document}

\begin{center}

{\Large \textbf{Martingale Posteriors for Discretely Observed Diffusions}}

\vspace{0.5cm}

BY  JINGNING YAO,  AJAY JASRA \& SHENG JIANG

{\footnotesize School of Data Science,  The Chinese University of Hong Kong,  Shenzhen,  Shenzhen, CN.}\\
{\footnotesize E-Mail:\,} \texttt{\emph{\footnotesize  jingningyao@link.cuhk.edu.cn; 
ajayjasra@cuhk.edu.cn; jiangsheng@cuhk.edu.cn
}}

\end{center}

\begin{abstract}
In this paper we consider parameter estimation for discretely observed diffusion processes.  In particular,  we focus on data that are observed at low frequency and methodology that can estimate parameters with uncertainty quantification.  Most statistical work in this domain develops advanced Markov chain Monte Carlo (MCMC) algorithms for sampling from the posterior of the parameters,  a task which is often complicated by the fact that one seldom has access to the transition density of the diffusion process; one has to combine sophisticated MCMC methods which are robust to the required time discretization of the diffusion, which can yield expensive algorithms.   We focus on developing the martingale posterior method \cite{fong} for the context of interest,  when one can only numerically approximate the transition density of the diffusion.  Based on using types of diffusion bridges (e.g.~\cite{schauer,vd_meulen_guided_mcmc}) we introduce a new martingale posterior method for parameter estimation for discretely observed diffusion processes.  We prove that this algorithm approximates,  in some sense,  the martingale posterior which has no time-discretization bias up-to $\mathcal{O}(\Delta)$ if $\Delta$ is the time discretization step.  Our approach is illustrated on several examples,  showing orders of magnitude speed up versus state-of-the-art MCMC algorithms.
\\
\noindent\textbf{Keywords}:  Uncertainty Quantification, Diffusion Processes, High and Low Frequency Observations.
\end{abstract}

\section{Introduction}

We are given a diffusion process for $X_0\in\mathbb{R}^d$ given:
\begin{equation}\label{eq:diff}
dX_t = \mu_{\theta}(X_t)dt + \sigma_{\theta}(X_t) dW_t
\end{equation}
where $X_t\in\mathbb{R}^d$,  $\theta\in\Theta=\mathbb{R}^{d_{\theta}}$,   $\mu_{\theta}:\mathbb{R}^d\rightarrow\mathbb{R}^d$,  
$\sigma_{\theta}:\mathbb{R}^d\rightarrow\mathbb{R}^d$ and $\{W_t\}_{t\geq 0}$ is standard Brownian motion.
We are given access to regularly observed data associated to the diffusion (e.g.~$x_1,\dots,x_T$) and we want to estimate $\theta$ with access to uncertainty quantification,  for example  in a Bayesian manner. We assume the model is uniformly elliptic and that there exists a Lebesgue transition density whose gradient in $\theta$ also exists.  
Diffusion models find numerous applications in real scientific problems such as finance,  econometrics,  biology and physics;  we refer the reader to the works \cite{bridges,beskos,beskos_ea,golightly,non_synch,roberts} and the references therein for further descriptions.  

When one seeks to estimate the parameters $\theta$,  with some measure of uncertainty quantification,  it is often the case that many researchers use a Bayesian perspective.  That is, to place a prior on $\theta$ and seek to compute expectations w.r.t.~the posterior distribution of the parameters conditional upon the observed data.  In most cases of practical interest this has two major issues, which are that (1) the posterior density is known,  at best, up-to a normalizing constant and (2) indeed the un-normalized posterior is only known up-to an approximation of the posterior density.  The latter is that whilst under fairly mild conditions \eqref{eq:diff} admits an almost sure unique solution and under further assumptions a Lebesgue transition density,  this density is not known and can (often) only be approximated,  in some way,  using numerical time-discretization methods,  such as Euler-Maruyama (e.g.~\cite{kloeden}).  Then one uses the approximated transition density to compute an un-normalized (approximation) of the posterior density and this is used as part of a Markov chain Monte Carlo (MCMC) algorithm.  By now there are numerous examples of such MCMC algorithms of which a non-exhaustive list includes \cite{bridges,golightly,jasra,non_synch,roberts,schauer,vd_meulen_guided_mcmc}.  These MCMC algorithms can be quite sophisticated,  using a combination of MCMC,  sequential Monte Carlo and sometimes even multilevel Monte Carlo \cite{giles,ml_rev}.

In terms of application,  the afore-mentioned MCMC algorithms can be quite expensive,  in that they constitute many complicated approaches and hence require specific codes that can be challenging to develop,  even with AI tools.  The methods also require long training and run periods which can be prohibative,  especially when the observed time series is quite long.  The objective of this work is to investigate faster methods for parameter estimation for diffusions with uncertainty quantification and specifically how they can be developed and implemented in a coherent manner.  The particular methodology that we focus upon is that of the martingale posterior \cite{cui,fong,mart1}, see also later work \cite{fong1,mgpn,ng},   who present an approach for uncertainty quantification based upon constructing an appropriate martingale.  This martingale can be quickly simulated forwards in time (based on model assumptions) and relying on the well-known martingale convergence theorem: there exists a limiting (in time) random variable. This limiting random variable can be simulated  using the afore-mentioned forward simulations.  In many experiments the authors \cite{cui,fong,mart1} show promising results in terms of fast simulation,  especially relative to MCMC.  The perspective of this paper is to focus on the score-based martingale methodology of \cite{cui} and develop new and novel approaches to implement it for diffusions. This latter method focusses upon the $\theta-$gradient of log of the observation density (the score).
We remark that this paper is not a philosophical one from the statistical perspective; we assume that the reader is convinced at the utility of the martingale posterior method and we then proceed to develop the said ideas for diffusion models. 

In this article we focus on low frequency observations,  that is,  that are separated by $\mathcal{O}(1)$ time.  In such a scenario,  it is known that, if one needs access to the transition density of the diffusion  as we do,  one must carefully construct an approximation of the density as we now explain.
If one considers the Euler-Maruyama time discretization,  then one can show that (for example) the approximated transition density over unit time is an integral over the time-steps of the diffusion between observation times.  It is well known that if say the time step is $\Delta_l=2^{-l}$ it can often be the case that an estimate of the transition density explodes exponentially fast in $l$; see for example \cite{durham} and we explain this point carefully in Section \ref{sec:method}.  In this paper the problem is even more accute as one needs access to the score function.
In this paper,  based on the works in \cite{schauer,vd_meulen_guided_mcmc} and its extensions \cite{bridges,beskos,non_synch} we develop a new score-based martingale posterior method for diffusions that are observed at low frequencies.  The main reason why one can escape the exponential variance explosion alluded to above,  is that one can, in principle,  define a martingale posterior in true continuous time.  Of course, in practice, this martingale posterior cannot be simulated and must be time discretized.  We then prove,  under mathematical assumptions that the expected square $L_2-$distance between the martingale posterior in continuous and discrete time is $\mathcal{O}(\Delta_l)$ hence confirming that the exponential variance in $l$ is avoided.   Finally we illustrate our methodology in several examples,  showing the improvement in terms of computational time versus some state-of-the-art MCMC algorithms.  

This article is structured as follows.  
In Section \ref{sec:review} we give a short review of the score-based martingale methodology which we shall adapt for diffusions.
In Section \ref{sec:method} we describe our methodology.  In Section \ref{sec:theory} our theoretical result are presented.  In Section \ref{sec:numerics} we investigate our methodology for several diffusions.  In Section \ref{sec:sum} the article is concluded and several avenues for future work are discussed. The mathematical proof of our main result and its assumptions are housed in Appendix \ref{app:theory}.

\section{Review}\label{sec:review}

\subsection{Basic Score-Based Martingale Posterior Approach}

Our approach is based on the ideas in \cite{cui} which we now recapitulate for completeness.  For each parameter $\theta\in\Theta=\mathbb{R}^{d_{\theta}}$,  let $F_{\theta}$ be a distribution  on $(\mathbb{R}^d,\mathscr{B}(\mathbb{R}^d))$ with $\theta-$differentiable Lebesgue density $f_{\theta}$.  Let $\theta_0\in\Theta$ be given and define the recursion
\begin{equation}\label{eq:sbm_rec}
\theta_k = \theta_{k-1} + \gamma_k \nabla \log\left\{f_{\theta_{k-1}}(X_k)\right\}
\end{equation}
where for each $k\geq 1$:
$$
X_k|\theta_0, \theta_1,x_1,\dots,\theta_{k-1},x_{k-1} \sim F_{\theta_{k-1}}
$$
with $x_0$ the null vector, 
$\gamma_k\in\mathbb{R}^+$ and $\sum_{k\geq 1} \gamma_k = \infty$, $\sum_{k\geq 1} \gamma_k^2< \infty$ is the step-size.  We remark that in this paper we only consider parameters on $\mathbb{R}^{d_{\theta}}$,  that is,  on unconstrained spaces.  

Denote by $\mathbb{E}$ the expectation operator associated to the afore-mentioned process $\theta_0, \theta_1,X_1,\dots$.
Set $\mathscr{F}_0$ as the trivial sigma field and for $k\geq 1$, 
$\mathscr{F}_k$ as the $\sigma-$field generated by $(\theta_0, \theta_1,X_1,\dots,\theta_k,X_k)$ if one can establish that
$\sup_{k\geq 0}\mathbb{E}[|\theta_k|]<+\infty$ ($|\cdot|$ is the $L_1-$norm for vectors) 
and that 
$$
\mathbb{E}\left[\sup_{\theta\in\Theta}\left|\nabla \log\left\{f_{\theta}(X_k)\right\}\right|\right] <+\infty
$$
then one has that for any $k\geq 1$
\begin{eqnarray*}
\mathbb{E}\left[\theta_k|\mathscr{F}_{k-1}\right] & = & \theta_{k-1} + \gamma_k
\mathbb{E}\left[\nabla \log\left\{f_{\theta_{k-1}}(X_k)\right\}|\mathscr{F}_{k-1}\right] \\
& = & \theta_{k-1}
\end{eqnarray*}
where we have assumed that it is legitimate to reverse the order of integration and differentiation and used the well-known property of the expectation of the score function.   We have thus established that $(\theta_k)_{k\geq 0}$ is a martingale relative to $(\mathscr{F}_k)_{k\geq 0}$.   Then by the martingale convergence theorem 
(e.g.~\cite[Section 1.3]{hall})
there exists a random variable $\theta^{\star}$ that is integrable (i.e.~$\mathbb{E}[|\theta^{\star}|]<+\infty$) such that we have,  almost surely
\begin{equation}\label{eq:mart_conv}
\lim_{k\rightarrow\infty}\theta_k  = \theta^{\star}.
\end{equation}
We note that, writing the $j^{\text{th}}-$element of a $d_{\theta}-$vector $x$ as $x^{(j)}$
$j\in\{1,\dots,d_{\theta}\}$ and supposing $\mathbb{E}[(\theta_0^{(j)})^2]<+\infty$
\begin{eqnarray*}
\mathbb{E}\left[\left(\theta^{\star,(j)}\right)^2\right]^{\tfrac{1}{2}} & = & \mathbb{E}\left[
\left(
\theta_{0}^{(j)} + \sum_{p=1}^{\infty}\gamma_k\nabla \log\left\{f_{\theta_{k-1}}(X_k)\right\}^{(j)}
\right)^2\right]^{\tfrac{1}{2}} \\
& \leq & C\left(1+
\mathbb{E}\left[\sum_{p=1}^{\infty}\gamma_p^2\left(\nabla \log\left\{f_{\theta_{k-1}}(X_k)\right\}^{(j)}\right)^2\right]
^{\tfrac{1}{2}}\right)
\end{eqnarray*}
where $C<+\infty$ is a finite constant.
Then if 
$$
\mathbb{E}\left[\sup_{\theta\in\Theta}\left(\nabla \log\left\{f_{\theta}(X_k)\right\}^{(j)}\right)^2\right] < C
$$
for each $j\in\{1,\dots,d_{\theta}\}$ and some $C<+\infty$ one has that 
$$
\mathbb{E}\left[\left(\theta^{\star,(j)}\right)^2\right]^{\tfrac{1}{2}} <+\infty
$$
as $\sum_{k\geq 1}\gamma_k^2<+\infty$.  This is enough to show that  $\theta^{\star}$ has finite variance.

We can interpret \eqref{eq:mart_conv} as follows: if one runs the recursion \eqref{eq:sbm_rec} indefinitely, then under the stated assumptions, there is a limiting random variable which can be used to represent an uncertainty about the model parameter.  In practice one does not know this random variable,  but for many interesting models,  one can sample \eqref{eq:sbm_rec} for a long time horizon  and many times; the resulting collection of samples can be used an approximation of $\theta^{\star}$.
As has been discussed in \cite{cui,fong} this is not only an embarrassingly parallel sampling method,  but is cheaper than many modern simulation techniques such as MCMC.  We refer the reader to the afore-mentioned references for the justification method, beyond the mathematical rationale that we have given. 

\subsection{Initializing the Process}

The approach just discussed does not leverage any observed data in any way,  which is clearly not a sensible idea in practice.  In this article we consider an initial training phase,  which can also be used to set the step-size (see below) by running the recursion \eqref{eq:sbm_rec},  except one does not sample the data and uses the observed data in its place; this was adopted by \cite{cui}.  We note that a sequence of parameters generated in such a manner \emph{would not} necessarily produce a martingale sequence as we now explain.  We remark however,  neither does this matter as the proposed idea is simply used to obtain $\theta_0$; the martingale property can be gained from the susbequent simulation of the recursion \eqref{eq:sbm_rec}.

Suppose,  for exposition purposes,  the data have been generated from model with the parameter also generated from some model,  say the generated value is $\theta^*$; write the associated expectation operator as $\mathbb{E}_{\theta^*}$ and the $\sigma-$field that has been generated by the data and parameter up-to data-point $k$ as $\mathscr{G}_k$ (set $\mathscr{G}_0=\mathscr{F}_0$).  Now if one considers the conditonal expection
\begin{eqnarray}
\mathbb{E}_{\theta^*}\left[\theta_k|\mathscr{G}_{k-1}\right] & = & \theta_{k-1} + \gamma_k
\mathbb{E}_{\theta^*}\left[\nabla \log\left\{f_{\theta_{k-1}}(X_k)\right\}|\mathscr{G}_{k-1}\right] \nonumber\\
& = & \theta_{k-1} + \int_{\mathbb{R}^d} \nabla \log\left\{f_{\theta_{k-1}}(x_k)\right\} f_{\theta^*}(x_k)dx_k\label{eq:no_mart_exp}
\end{eqnarray}
one would not have a martingale in general,  even if the second term on the right hand side of \eqref{eq:no_mart_exp} is finite.  It may be possible to obtain a super or sub martingale,  depending on the properties of $\nabla \log\left\{f_{\theta_{k-1}}(x_k)\right\}$ but this is not useful here. 
As stated above this lack of martingale property is not necessarily an issue as we initialize the algorithm using this scheme and moreover one can use this phase to train the step-size adaptively as is often used in stochastic gradient/ascent methods.  It should be stated if the realized data are assumed to be generated from the process described in \eqref{eq:sbm_rec} one would of course obtain the martingale property,  but as we have stated it is not needed and is  computationally prohibitive.

\subsection{On the Step-Size and Pre-Conditioners}

In much of our preliminary investigations we have found that the step-size is an important determining factor in the associated limiting random variable $\theta^{\star}$.  The link between stochastic gradient ascent and \eqref{eq:sbm_rec} is relatively clear,  so it makes sense that one should try and utilize ideas from that literature to enhance the score-based martingale approach.   One method is the notion of using a pre-conditioner (e.g.~\cite{li,xli}):
$$
\theta_k = \theta_{k-1} + \gamma_k P_k\nabla \log\left\{f_{\theta_{k-1}}(X_k)\right\}
$$
for some $d_{\theta}\times d_{\theta}$ matrix $P_k$.  When available,  at least in convex optimization,  one often uses
$P_k=H_{\theta_{k-1}}^{-1}$ where $H_\theta = \nabla^2 \log\{f_{\theta}(x)\}$.  This choice is unavailable to us,  in the sense that even if the hessian $H_\theta$ is available,  it would destroy the martingale property.  We note that \cite{cui} suggest using the square root of the Fisher information matrix as a pre-conditioner,  but again this is seldom available.
In some contexts e.g.~\cite{xli},  one can try adaptively learn a pre-conditioner and this could be used in the data learning phase of the algorithm.  In this article we focus on learning a `good' step-size in the afore-mentioned phase and this is discussed in our numerical examples.   We note the parallel with stochastic optimization also suggests a wide variety of strategies (such as sample averaging) that could be investigated for the score-based martingale approach,  but our focus is on deriving new algorithms for diffusions and this is what we describe now.

\section{Score-Based Martingales for Diffusions}\label{sec:method}

\subsection{Assumptions}

We now return to the diffusion process \eqref{eq:diff} and we will assume the following.
\begin{itemize}
                \item[(i)] For each $\theta\in\Theta$,  
\begin{itemize}
\item{each element of $\sigma_{\theta}$, is twice continuously differentiable
with bounded derivatives and globally Lipschitz}
\item{$\Sigma_{\theta}(x)=\sigma_{\theta}(x)\sigma_{\theta}(x)^{\top}$ is uniformly elliptic.}
\end{itemize}
                \item[(ii)] For each $\theta$,  each element of $\mu_{\theta}$ is twice continuously differentiable with bounded derivatives and globally Lipschitz.
        \end{itemize}
We note that one could relax several assumptions,  but for ease of exposition we do not do this.  Assumptions of these type have been used for several MCMC methods such as \cite{bridges,ub_grad,non_synch} and some discussion is given there.  We remark that the uniformly elliptic assumption allows us to use fairly simple diffusion approximations (such as the Euler-Maruyama method) and in general most of the approximation methods would work well.   The assumptions also ensure the existence of Lebesgue transition density which,  between two times $0<s<t$
and for any fixed $(x_s,x_t)\in\mathbb{R}^{2d}$ is denoted as
as $f_{\theta,s,t}(x_t|x_s)$.  Without further discussion it is assumed that the gradient vector in $\theta$ of
$\log\left\{f_{\theta,s,t}(x_t|x_s)\right\}$ is well-defined. 

\subsection{Exact Method}

We will now consider a type of score-based martingale method in the following context.  We shall suppose access to
observations $x_0,x_1,\dots,x_T$ of the diffusion process.  The notation is such  that the data are observed at regular and unit times,  however,  this convention is made for presentational purposes only.  The approaches to be detailed can handle data at a collection of observation times $t_0,t_1,\dots,t_T$ say with $0\leq t_0<t_1<\cdots<T_T$,  which includes irregular observation times and in principle high-frequency observation times.  We note,  however,  that the subsequent ideas are developed precisley for the case of low-frequency data,  i.e.~$\mathcal{O}(1)$ time between each data point;  the high frequency case is discussed in Section \ref{sec:sum} as part of future work.

As we now consider a time dependent process,  as opposed to an i.i.d.~type process in the previous section,  one must be more careful in order to construct a martingale.   The initial phase of our approach is as follows:
\begin{enumerate}
\item{Set $\theta_0\in\Theta$.}
\item{For $k\in\{1,\dots,T\}$,  set
$$
\theta_k = \theta_{k-1} + \gamma_k \nabla \log\left\{f_{\theta,k-1,k}(x_k|x_{k-1})\right\}.
$$
}
\end{enumerate}
In practice,  within this phase,  one could use pre-conditioners or adaptive methods to learn the step-size.  As we intend to analyze this `vanilla' method mathematically,  this is not added to the exposition here.   In terms of random variables,  implicitly, there is \emph{some} data generating process for $X_0,\dots,X_T$,  for instance the exact diffusion process with some given $\theta^*\in\Theta$.  We do not make any assumption about this process other than what will be needed below. 

Let $(\theta,x)\in\Theta\times\mathbb{R}^d$,  $0<s<t$ be given and denote by $F_{\theta,s,t}(\cdot|x)$ as the distribution associated to $f_{\theta,s,t}(\cdot|x)$.
The next phase is run for `long-time',  so that we do not specify a stopping time,  is generated as follows. 
\begin{enumerate}
\item{Input $\theta_T$ and $X_{T}$ from the first phase. Set $k=T+1$.}
\item{Sample $X_{k}|x_{k-1},\theta_{k-1}\sim F_{\theta_{k-1},k-1,k}(\cdot|x_{k-1})$.}
\item{Update
$$
\theta_k = \theta_{k-1} + \gamma_k \nabla \log\left\{f_{\theta,k-1,k}(x_k|x_{k-1})\right\}
$$
and $k=k+1$.  Go to the start of 2..}
\end{enumerate}
To construct our martingale,  we specify that $\mathscr{F}_T$ is the $\sigma-$field that is generated by
$(\theta_0,X_0,\dots,\theta_T,X_T)$ and for any $t\geq T+1$ we set
 $\mathscr{F}_t$ as the $\sigma-$field that is generated by $(\theta_0,X_0,\dots,\theta_t,X_t)$.    Write the expectation associated to the law of the process described in phases 1 and 2 as
$\mathbb{E}$ and suppose that 
\begin{eqnarray*}
\sup_{k\geq 0}\mathbb{E}[|\theta_k|] & <& +\infty \\
\sup_{k\geq 1}\mathbb{E}\left[\sup_{\theta\in\Theta}\left|\nabla \log\left\{f_{\theta,k-1,k}(X_k|X_{k-1})\right\}\right|\right] & < & +\infty
\end{eqnarray*}
then we have that $(\theta_k)_{k\geq T}$ is a martingale relative to $(\mathscr{F}_k)_{k\geq T}$.  One is then in a position to consider uncertainty quantification associated to the parameters of the diffusion.  

\subsubsection{Problems with the Exact Method}\label{sec:td_issues}

The main problem of using the exact method is the fact that one needs to be able to evaluate 
$\nabla\log\left\{f_{\theta,s,t}(x_t|x_s)\right\}$  and sample from the exact transition law of the diffusion process.  
This is only possible for a very small class of diffusions such as the Ornstein-Uhlenbeck process.
In terms of exact simulation of diffusions there are several methods in the literature (e.g.~\cite{beskos_ea,blanchet})  but typically these are only
useful if $d$ (the dimension of the diffusion) is small or one does not need to repeatedly simulate the diffusion.
One could also seek to use the exact simulation methods to produce unbiased estimates of $\nabla\log\left\{f_{\theta,s,t}(x_t|x_s)\right\}$,  but this would again suffer the same problem as just mentioned.  

One of the main approaches to deal with the main issues of the exact method is to adopt time-discretizations,  to approximate the law of the diffusion process.  Let $l\in\mathbb{N}$ be given and set $\Delta_l=2^{-l}$ then one of the most simple time discretization methods is the Euler-Maruyama approach for $(k,j)\in\mathbb{N}\times\{0,\dots,\Delta_l^{-1}-1\}$:
\begin{equation}\label{eq:em_approx}
X_{k-1+(j+1)\Delta_l} = X_{k-1+j\Delta_l} + \mu_{\theta}(X_{k-1+j\Delta_l})\Delta_l + 
\sigma_{\theta}(X_{k-1+j\Delta_l})\left[W_{k-1+(j+1)\Delta_l}-W_{k-1+j\Delta_l}\right].
\end{equation}
In such a scenario,  one can sample this process,  assuming one can evaluate the drift and diffusion coefficients and the transition density over time unit $\Delta_l$ is known.   Moreover,  the convergence of the scheme as $l\rightarrow+\infty$ has been well studied and we refer the reader to \cite{kloeden} and the references therein. 

The basic idea is then to replace the exact simulation of the diffusion with the Euler-Maruyama simulation.
However,  whilst the transition density over time unit $\Delta_l$  is known,  it is typically the case that the transition density over time 1 (say) is unknown,  due to the fact that it is written in terms of an integral over the path $x_{k-1+\Delta_l}, \dots, x_{k-\Delta_l}$.  Moreover,  if one tries to obtain the gradient of the log of the transition density obtaining (say) an unbiased estimate,  the variance can explode as $l$ grows.   Note that conditional on $x_{k-1},x_k$ unbiasedness is needed so as to obtain the martingale property.
To understand the  point
on variance explosion,  denote the transition density associated to \eqref{eq:em_approx} as $f_{\theta,\Delta_l}(x'|x)$ then the transition density over unit time is
$$
f_{\theta,k-1,k}^l(x_k|x_{k-1}) := 
\int_{\mathbb{R}^{d(\Delta_l^{-1}-1)}} \left\{\prod_{j=0}^{\Delta_l^{-1}-1} 
f_{\theta,\Delta_l}(x_{k-1+(j+1)\Delta_l}|x_{k-1+j\Delta_l})\right\} dx_{k-1+\Delta_l}\dots dx_{k-\Delta_l}.
$$
Then under standard regularity assumptions,  which are omitted for brevity,  one has that
$$
\nabla \log\left\{f_{\theta,k-1,k}^l(x_k|x_{k-1})\right\} = 
$$
$$
\int_{\mathbb{R}^{d(\Delta_l^{-1}-1)}} \nabla \log \left\{\prod_{j=0}^{\Delta_l^{-1}-1} 
f_{\theta,\Delta_l}(x_{k-1+(j+1)\Delta_l}|x_{k-1+j\Delta_l})\right\} \pi_{\theta}^l(x_{k-1+\Delta_l},\dots,x_{k-\Delta_l}
|x_{k-1},x_k) dx_{k-1+\Delta_l}\dots dx_{k-\Delta_l}
$$
where
$$
\pi_{\theta}^l(x_{k-1+\Delta_l},\dots,x_{k-\Delta_l}
|x_{k-1},x_k) = \frac{\prod_{j=0}^{\Delta_l^{-1}-1} 
f_{\theta,\Delta_l}(x_{k-1+(j+1)\Delta_l}|x_{k-1+j\Delta_l})}{f_{\theta,k-1,k}^l(x_k|x_{k-1})}.
$$
So if one can sample from $\pi_{\theta}^l$ one could produce an unbiased estimate of $\nabla \log\left\{f_{\theta,k-1,k}^l(x_k|x_{k-1})\right\}$ and use this as part of an approximation of the exact method.  However,  estimates based on the mentioned idea often have a variance that will explode as $l$ grows and this is well understood in the MCMC and importance sampling literature;  see for example \cite{durham}. 
Intuitively this is because as $l$ grows the number of terms in the product grows and one does not have a convergence of this term (at least not fast enough). 
This variance explosion is an issue as one would like to run an algorithm which has $l$ large so as to approximate as closely as possible the $\theta^{\star}$ that is produced from the exact method.  

As has been realized in several works \cite{bridges,beskos,ub_grad,non_synch} if one seeks to estimate
$\nabla\log\left\{f_{\theta,k-1,k}(x'|x)\right\}$ one must start with a formula in continuous time and then seek a time discretized approximation of the formula.   In \cite{ub_grad} the authors focus on using the Girsanov change of measure,  but this is overly restrictive as it does not allow dependence of $\theta$ in the diffusion coefficient.  The approaches in 
\cite{bridges,beskos,non_synch} use an adaptation of a type of diffusion bridge method that has been considered in \cite{schauer,vd_meulen_guided_mcmc} and we now review that idea,  so that  it can be utiltized for our score-based martingale method.

\subsection{Approximating $\nabla\log\left\{f_{\theta,k-1,k}(x'|x)\right\}$}

The following exposition follows the reviews in \cite{bridges,beskos,non_synch}.
Consider the case $t\in[s_1,s_2]$, $0\leq s_1,s_2\leq T$ and let
$\mathbf{X}_{[s_1,s_2]}:=\{X_{t}\}_{t\in[s_1,s_2]}$, and $\mathbf{W}_{[s_1,s_2]}:=\{W_{t}\}_{t\in[s_1,s_2]}$.  Set $f_{\theta,t,s_{2}}(x'|x)$ denote the unknown transition density from time $t$ to $s_2$ associated to \eqref{eq:diff}.
If one could sample from $f_{\theta,s_1,s_2}$ to obtain $(x, x')\in \mathbb{R}^{2d}$, we will explain that 
we can interpolate these points by using a bridge process.
This process will have
a drift given by $\mu_{\theta}(x)+\Sigma_{\theta}(x)\nabla_x\log{f}_{\theta,t,s_2}(x'|x)$.
Let $\mathbb{P}_{\theta,x,x'}$ denote the law of the solution of the stochastic differential equation (SDE) \eqref{eq:diff}, on $[s_1,s_2]$, started at $x$ and conditioned to hit $x'$ at time $s_2$.

We introduce a user-specified auxiliary process $\{\tilde{X}_t\}_{t\in[s_1,s_2]}$ following:
\begin{align}
\label{eq:aux_SDE_tilde}
d \tilde X_{t} = \tilde \mu_{\theta}(t,\tilde X_{t})dt + \tilde \sigma_{\theta}(t,\tilde X_{t})dW_{t}, \quad t\in[s_1,s_2],\quad~\tilde{X}_{s_1} =x, 
\end{align}
where for each $\theta\in\Theta$, $\tilde \mu_{\theta}:[s_1,s_2]\times\mathbb{R}^{d}\rightarrow\mathbb{R}^{d}$ and $\tilde \sigma_{\theta}:\mathbb{R}^{d}\rightarrow\mathbb{R}^{d\times d}$ is
such that for each $\theta\in\Theta$
$$
\tilde \Sigma_{\theta} (s_2,x'):= \tilde \sigma_{\theta}(s_2,x') \tilde \sigma_{\theta}(s_2,x')^{\top} \equiv \Sigma_{\theta}(x').
$$ 
\eqref{eq:aux_SDE_tilde} is specified so that its transition density $\tilde{f}_\theta$ is available; see \cite[Section 2.2]{schauer} for the technical conditions on $\tilde \mu_{\theta}, \tilde \Sigma_{\theta}, \tilde f_\theta$. 
The main purpose of $\{\tilde X_t\}_{t\in[s_1,s_2]}$ is to sample $x'$ and use its transition density to construct another process $\{X_t^\circ\}_{t\in[s_1,s_2]}$ conditioned to hit $x'$ at $t=s_2$; which in turn will be an importance proposal for $\{X_t\}_{t\in[s_1,s_2]}$.  Set:
\begin{align}
\label{eq:aux_SDE}
d X^\circ_{t} = \mu_{\theta,s_2}^{\circ}(t,X^{\circ}_{t}; x')dt + \sigma_{\theta}(X^{\circ}_{t})dW_{t}, \quad t\in[s_1,s_2],\quad~X^{\circ}_{s_1} =x, 
\end{align}
where:
$$
\mu_{\theta,s_2}^{\circ}(t,x;x')=\mu_{\theta}(x)+\Sigma_{\theta}(x)\nabla_x\log\tilde{f}_{\theta,t,s_2}(x'|x),
$$ 
and denote by
 $\mathbb{P}^\circ_{\theta,x,x'}$ the probability law of the solution of (\ref{eq:aux_SDE}). 
The SDE in \eqref{eq:aux_SDE} yields: 
\begin{equation}
\mathbf{W}\rightarrow C_{\theta,s_1,s_2}(x,\mathbf{W}_{[s_1,s_2]},x'),
\label{eq:map}
\end{equation} 
mapping the driving Wiener noise $\mathbf{W}$ to the solution of \eqref{eq:aux_SDE}, reparametering the problem from $\mathbf{X}$ to $(\mathbf{W},x')$.

\cite{schauer} prove that $\mathbb{P}_{\theta,x,x'}$ and $\mathbb{P}^\circ_{\theta,x,x'}$ are absolutely continuous w.r.t.~each other, with Radon-Nikodym derivative: 
\begin{equation}
\frac{d\mathbb{P}_{\theta,x,x'} }{d \mathbb{P}^\circ_{\theta,x,x'} }(\mathbf{X}_{[s_1,s_2]})=
\exp\Big\{  \int_{s_1}^{s_2} L_{\theta,s_2}(t,X_t)dt\Big\} \times \frac{ \tilde{f}_{\theta,s_1,s_2}(x'|x)}{f_{\theta,s_1,s_2}(x'|x)},
\label{eq:density}
\end{equation}
where: 
\begin{align*}
L_{\theta,s_2}&(t,x):=\left(\mu_\theta(x)- \tilde{\mu}_\theta(t,x)\right)^{\top}\, \nabla_x\log\tilde{f}_{\theta,t,s_2}(x'|x)\\
&-\frac{1}{2}\textrm{Tr}\,\Big\{\,\big[ \Sigma_{\theta}(x)-\tilde{\Sigma}_{\theta}(t,x)\big] \big[ -\nabla_x^2\log\tilde{f}_{\theta,t,s_2}(x'|x)-\nabla_x\log\tilde{f}_{\theta,t,s_2}(x'|x)\nabla_x\log\tilde{f}_{\theta,t,s_2}(x'|x)^{\top} \big]\,\Big\} ,         
\end{align*}
with $\textrm{Tr}(\cdot)$ denoting the trace of a squared matrix.  
Set
$$
R_{\theta,s_1,s_2}(\mathbf{X}_{[s_1,s_2]}) := \frac{d\mathbb{P}_{\theta,x,x'} }{d \mathbb{P}^\circ_{\theta,x,x'} }(\mathbf{X}_{[s_1,s_2]})f_{\theta,s_1,s_2}(x'|x)
$$
Let $\mathbb{W}$ be the law of Brownian motion with associated expectation operator $\mathbb{E}$
then
$$
f_{\theta,k-1,k}(x'|x) = \mathbb{E}\left
[R_{\theta,k-1,k}\left(
C_{\theta,k-1,k}(x,\mathbf{W}_{[k-1,k]},x')\right)\right].
$$

\subsection{Time Discretization}

We first give the standard Euler-Maruyama time discretization of the solution to \eqref{eq:aux_SDE} (associated to a time interval $[k-1,k]$, $k\in\{1,\dots,T\}$) on a regular grid of spacing $\Delta_{l}=2^{-l}$, with starting point $x_{k-1}^{\circ}$ and ending point $x_{k}^{\circ}$. That is for $j\in\{0,1\dots,\Delta_{l}^{-1}-2\}$:
\begin{align}\label{eq:disc_circ}
X_{k-1+(j+1)\Delta_{l}}^{\circ} & = X_{k-1+j\Delta_{l}}^{\circ} + 
\mu_{\theta,k}^{\circ}(k-1+j\Delta_{l},X^{\circ}_{k-1+j\Delta_{l}}; x_{k}^\circ)\Delta_{l} + 
\nonumber \\ &
\sigma_{\theta}(X^{\circ}_{k-1+j\Delta_{l}})\left[W_{k-1+(j+1)\Delta_{l}}-W_{k-1+j\Delta_{l}}\right].
\end{align}
Given $(x_{k-1}^{\circ},x_{k}^{\circ})$ and $\mathbf{W}_{[k-1,k]}^l=(W_{k-1+\Delta_{l}}-W_{k-1},\dots,W_{k-\Delta_l}-W_{k-2\Delta_{l}})$,  the recursion \eqref{eq:disc_circ} induces a discretized path 
$X_{k-1+\Delta_{l}}^{\circ},\dots,X_{k-\Delta_{l}}^{\circ}$ and we write such a path, including the starting and ending points with the notation
$$
C_{\theta,k-1,k}^l(x_{k-1}^{\circ},\mathbf{W}_{[k-1,k]}^l,x_{k}^{\circ}).
$$
We also need a discretization of the Radon-Nikodym derivative.  Consider 
$$
\mathbf{X}_{[k-1,k]}^l=(X_{k-1},X_{k-1+\Delta_{l}},\dots,X_k)
$$ 
then we set
$$
R_{\theta,k-1,k}^l(\mathbf{X}_{[k-1,k]}^l) := 
\exp\Big\{ \sum_{j=0}^{\Delta_{l}^{-1}-1} 
L_{\theta,k}(t,X_{k-1+j\Delta_{l}})\Delta_{l}\Big\} \tilde{f}_{\theta,k-1,k}(X_{k}|X_{k-1}).
$$
We will work with
$$
f_{\theta,k-1,k}^l(x'|x) \propto \mathbb{E}\left
[R_{\theta,k-1,k}^l\left(
C_{\theta,k-1,k}^l(x,\mathbf{W}_{[k-1,k]}^l,x')\right)\right].
$$

\subsection{Time Discretized Method}

We propose the following scheme to estimate the parameters.  The initial iterate $\theta_0^l$ is given.  The following method is called \textbf{MPD} in the subsequent discussion, for martingale posterior for diffusions.

The first phase is as follows.

\begin{enumerate}
\item{Set $k=1$ and go to 2..} 
\item{Iterate: 
\begin{itemize}
\item{Sample $\mathbf{W}_{[k-1,k]}^l|x_{k},x_{k-1}$ from
$$
\pi_{\theta_{k-1}^l}^l(d\mathbf{w}_{[k-1,k]}^l|x_{k},x_{k-1}) \propto
R_{\theta_{k-1}^l,k-1,k}^l\left(
C_{\theta_{k-1}^l,k-1,k}^l(x_{k-1},\mathbf{w}_{[k-1,k]}^l,x_{k})\right) \mathbb{W}(d\mathbf{w}_{[k-1,k]}^l)
$$
compute $\nabla \log\left(R_{\theta_{k-1}^l,k-1,k}^l\left(
C_{\theta_{k-1}^l,k-1,k}^l(x_{k-1},\mathbf{w}_{[k-1,k]}^l,x_k)\right)\right)$.}
\item{Sample $\mathbf{V}_{[k-1,k]}^l,u_k|x_{k-1}$
$$
\pi_{\theta_{k-1}^l}^l(d\mathbf{v}_{[k-1,k]}^l,u_k|x_{k-1}) \propto
R_{\theta_{k-1}^l,k-1,k}^l\left(
C_{\theta_{k-1}^l,k-1,k}^l(x_{k-1},\mathbf{v}_{[k-1,k]}^l,u_k)\right) \mathbb{W}(d\mathbf{v}_{[k-1,k]}^l)
$$
and compute $\nabla \log\left(R_{\theta_{k-1}^l,k-1,k}^l\left(
C_{\theta_{k-1}^l,k-1,k}^l(x_{k-1},\mathbf{v}_{[k-1,k]}^l,u_k)\right)\right)$.}
\item{Update:
\begin{align*}
\theta_k^l & = \theta_{k-1}^l + \gamma_k \Big\{
\nabla \log(R_{\theta_{k-1}^l,k-1,k}^l\left(
C_{\theta_{k-1}^l,k-1,k}^l(x_{k-1},\mathbf{w}_{[k-1,k]}^l,x_k)\right)) - \\
& \nabla \log(R_{\theta_{k-1}^l,k-1,k}^l\left(
C_{\theta_{k-1}^l,k-1,k}^l(x_{k-1},\mathbf{v}_{[k-1,k]}^l,u_k)\right))
\Big\}.
\end{align*}
}
\end{itemize}
Go to 3..
}
\item{Set $k=k+1$ and if $k=T+1$ stop otherwise go to 2..}
\end{enumerate}

The second phase is as follows.

\begin{enumerate}
\item{Input $\theta_T^l$ and $X_T$ from the first phase, set $k=T+1$.  Sample $X_k|x_{k-1}$ using  
$$
f_{\theta_{k-1}^l,k-1,k}^l(x_k|x_{k-1}) \propto \int R_{\theta_{k-1}^l,k-1,k}^l\left(
C_{\theta_{k-1}^l,k-1,k}^l(x_{k-1},\mathbf{w}_{[k-1,k]}^l,x_k)\right) \mathbb{W}(d\mathbf{w}_{[k-1,k]}^l)
$$
 Go to 2..} 
\item{Iterate: 
\begin{itemize}
\item{Sample $\mathbf{W}_{[k-1,k]}^l|x_{k},x_{k-1}$ from
$$
\pi_{\theta_{k-1}^l}^l(d\mathbf{w}_{[k-1,k]}^l|x_{k},x_{k-1}) \propto
R_{\theta_{k-1}^l,k-1,k}^l\left(
C_{\theta_{k-1}^l,k-1,k}^l(x_{k-1},\mathbf{w}_{[k-1,k]}^l,x_{k})\right) \mathbb{W}(d\mathbf{w}_{[k-1,k]}^l)
$$
compute $\nabla \log\left(R_{\theta_{k-1}^l,k-1,k}^l\left(
C_{\theta_{k-1}^l,k-1,k}^l(x_{k-1},\mathbf{w}_{[k-1,k]}^l,x_k)\right)\right)$.}
\item{Sample $\mathbf{V}_{[k-1,k]}^l,u_k|x_{k-1}$
$$
\pi_{\theta_{k-1}^l}^l(d\mathbf{v}_{[k-1,k]}^l,u_k|x_{k-1}) \propto
R_{\theta_{k-1}^l,k-1,k}^l\left(
C_{\theta_{k-1}^l,k-1,k}^l(x_{k-1},\mathbf{v}_{[k-1,k]}^l,u_k)\right) \mathbb{W}(d\mathbf{v}_{[k-1,k]}^l)
$$
and compute $\nabla \log\left(R_{\theta_{k-1}^l,k-1,k}^l\left(
C_{\theta_{k-1}^l,k-1,k}^l(x_{k-1},\mathbf{v}_{[k-1,k]}^l,u_k)\right)\right)$.}
\item{Update:
\begin{align*}
\theta_k^l & = \theta_{k-1}^l + \gamma_k \Big\{
\nabla \log(R_{\theta_{k-1}^l,k-1,k}^l\left(
C_{\theta_{k-1}^l,k-1,k}^l(x_{k-1},\mathbf{w}_{[k-1,k]}^l,x_k)\right)) - \\
& \nabla \log(R_{\theta_{k-1}^l,k-1,k}^l\left(
C_{\theta_{k-1}^l,k-1,k}^l(x_{k-1},\mathbf{v}_{[k-1,k]}^l,u_k)\right))
\Big\}.
\end{align*}
}
\end{itemize}
Go to 3..
}
\item{Set $k=k+1$ and sample $X_k|x_{k-1}$ using 
\begin{equation}\label{eq:f_disc}
f_{\theta_{k}^l,k-1,k}(x_k|x_{k-1}) \propto \int R_{\theta_k^l,k-1,k}^l\left(
C_{\theta_k^l,k-1,k}^l(x_{k-1},\mathbf{w}_{[k-1,k]}^l,x_k)\right) \mathbb{W}(d\mathbf{w}_{[k-1,k]}^l)
\end{equation}
Go to 2..}
\end{enumerate}

\subsubsection{On the Martingale Property}

To construct our martingale,  we specify that $\mathscr{F}_T$ is the $\sigma-$field that is generated by
$$
\mathsf{R}_T^l := \left(\theta_0,X_0,\dots,\theta_T^l,X_T, \mathbf{W}_{[0,1]}^l,\dots,\mathbf{W}_{[T-1,T]}^l,
\mathbf{V}_{[0,1]}^l,\dots,\mathbf{V}_{[T-1,T]}^l, U_1,\dots,U_T
\right).
$$
For any $t\geq T+1$ we set
 $\mathscr{F}_t$ as the $\sigma-$field that is generated by 
$$
\left(\mathsf{R}_T^l, \theta_{T+1}^l,X_{T+1},\dots, \theta_{t}^l,X_{t},
\mathbf{W}_{[T,T+1]}^l,\dots,\mathbf{W}_{[t-1,t]}^l,
\mathbf{V}_{[T,T+1]}^l,\dots,\mathbf{V}_{[t-1,t]}^l, U_{T+1},\dots,U_t
\right).
$$
Set
$$
\mathsf{S}_{k-1,k}^l(x_{k-1:k},u_k,\theta,\mathbf{w}_{[k-1,k]}^l,\mathbf{v}_{[k-1,k]}^l) = 
$$
\begin{equation}\label{eq:change_score}
\nabla \log(R_{\theta,k-1,k}^l\left(C_{\theta,k-1,k}^l(x_{k-1},\mathbf{w}_{[k-1,k]}^l,x_k)\right)) -
 \nabla \log(R_{\theta,k-1,k}^l\left(C_{\theta,k-1,k}^l(x_{k-1},\mathbf{v}_{[k-1,k]}^l,u_k)\right)).
\end{equation}
 Write the expectation associated to the law of the process described in phases 1 and 2 as
$\mathbb{E}$ and suppose that 
\begin{eqnarray*}
\sup_{k\geq 0}\mathbb{E}[|\theta_k|] & <& +\infty \\
\sup_{k\geq 1}\mathbb{E}\left[\sup_{\theta\in\Theta}
\left|\mathsf{S}_{k-1,k}^l(X_{k-1:k},U_k,\theta,,\mathbf{W}_{[k-1,k]}^l,\mathbf{V}_{[k-1,k]}^l)\right|\right] & < & +\infty
\end{eqnarray*}
then we have that $(\theta_k^l)_{k\geq T}$ is a martingale relative to $(\mathscr{F}_k)_{k\geq T}$.   We suppose that the martingale property holds,  without further technical discussion. 

\subsubsection{Comments on the Method and Practical Considerations}

One may wander why the expression $\mathsf{S}_{k-1,k}^l(x_{k-1:k},u_k,\theta,\mathbf{w}_{[k-1,k]}^l,\mathbf{v}_{[k-1,k]}^l)$ in \eqref{eq:change_score} is used in step 2,  the update of the parameter.  This is because 
the normalizing constant of 
$f_{\theta_{k}^l,k-1,k}(x_k|x_{k-1})$ in \eqref{eq:f_disc} is not typically one,  so must modify the score that one uses in the standard approach.  

In practice to compute $\nabla \log\left(R_{\theta,k-1,k}^l\left(
C_{\theta,k-1,k}^l(x_{k-1},\mathbf{w}_{[k-1,k]},x_k)\right)\right)$ one would need to use automatic differentiation as one does not have access to exact expressions.  This is what is done in all of our numerical simulations. 
We also note that the simulation from $\pi_{\theta_{k-1}^l}^l(d\mathbf{w}_{[k-1,k]}^l|x_{k},x_{k-1})$ and
$\pi_{\theta_{k-1}^l}^l(d\mathbf{v}_{[k-1,k]}^l,u_k|x_{k-1})$ is not standard and we have used rejection sampling in our examples.  As rejection sampling often only works well in low or moderate dimensions,  this suggests that the approach that we have outlined is useful, relative to MCMC, when $d$ (the dimension of the diffusion) is small and the length of the time series $T$ is large.  The cost of each update step is $\mathcal{O}(\Delta_l^{-1})$,  but this does not incorporate the cost of rejection sampling.  In our examples we have not found that the latter cost dominates the simulation.

\section{Theoretical Result}\label{sec:theory}

We now present our main theoretical result.  The objective is to compare the iterate $\theta_k^l$ from \textbf{MPD} to the one that is produced by the `perfect' version of \textbf{MPD} ($\theta_k$),  that is the one as $l\rightarrow\infty$; for completeness this is detailed in Appendix \ref{app:theory},  but it should be clear what that means on just reading 
\textbf{MPD}.  We note of course that the perfect version of \textbf{MPD} can never be implemented,  but would produce a martingale posterior that does not possess any time-discretization bias.  The objective of the following result is to understand the cost we pay versus this perfect \textbf{MPD} in terms of the time-discretization parameter.

We have the following result whose assumptions, (A\ref{ass:1}) and (A\ref{ass:2}),  are stated in Appendix \ref{app:theory_ass} and proved in Appendix \ref{app:theory_prf}.  We denote by $\|\cdot\|$ denotes the $L_2-$norm for vectors and the meaning of the expectation operator $\check{\mathbb{E}}$ below is clarified in Appendix \ref{app:theory_ass}.

\begin{theorem}\label{theo:main}
Assume (A\ref{ass:1}-\ref{ass:2}).  Then there exists a $C<+\infty$ such that for any $l\in\mathbb{N}$ we have
$$
\sup_{n\geq 1}\check{\mathbb{E}}\left[\left\|\theta_n^l-\theta_n\right\|^2\right] \leq C\Delta_l.
$$
\end{theorem}

The implication of this result is simply this.  We have shown that,  under assumptions,  in terms of the expected square $L_2-$norm,  that the \textbf{MPD} method maintains a $\mathcal{O}(\Delta_l)$ distance to the martingale posterior that has no time discretization bias.  In other words,  we avoid any possible exponential variance explosion in terms of the time discretization parameter $l$,  which as remarked in Section \ref{sec:td_issues} is a possibility if one does not carefully construct the approximation of the transitition density of the diffusion.  The main reason for this control is the fact that there is a martingale posterior in true continuous time and,  indeed,  that using Euler-Maruyama time discretizations one has a consistent,  in terms of $l$,  approximation of the said martingale posterior.  From a computational perspective this is rather reassuring in that one does not have to deal with exponential variances when improving the accuracy of the time discretization.  

\section{Numerical Simulations}\label{sec:numerics}

\subsection{Example 1}

We consider the Ornstein-Uhlenbeck (OU) process defined by the stochastic
differential equation
\begin{equation*}
  \mathrm{d}X_t = \theta(\mu - X_t)\,\mathrm{d}t + \sigma\,\mathrm{d}W_t,
\end{equation*}
where $\theta\in\mathbb{R}^+$ is the parameter of interest and $\mu,\sigma$ fixed.  We note that of course one
could run an exact method for this diffusion, but it is used as a test case.  For the auxiliary process we use another OU process except with $\theta$ fixed at some reference value,  and we call this parameter $ \theta_{\mathrm{aux}}$.

We will simulate 100 data points, $x_0 = 10$,  with the true parameters $\theta = 3$, $\mu = 10$, and $\sigma = 0.5$,
with only $\theta$ treated as unknown.  The data are generated at equally spaced times with gap 0.2.
The auxiliary process uses
$\theta_{\mathrm{aux}} = 5$, and each observation interval is subdivided
into $16$ Euler--Maruyama sub-steps of size $\Delta_l = 0.0125$.
The algorithm is initialised at $\theta_0 = 5$.  The stochastic gradient
update uses base learning rate $\eta = 30$ and offset $c = 50$,
\begin{equation}
  \gamma_k = \frac{30}{k + 50}.
\end{equation}
After the Phase~1 processes all $100$
observed  data points Phase~2 extends the trajectory by $300$ generative steps with the iteration counter continuing from $k = 100$.
The entire procedure is repeated  $100$ times with
independent random seeds, all sharing the same observed dataset.


Figure~\ref{fig:ou-result} displays the estimation trajectories from all
$100$ independent repetitions.  The blue curves correspond to Phase~1
(iterations $0$--$100$), where the estimator is driven by the observed
data; the red curves correspond to Phase~2 (iterations $100$--$400$),
where it continues with generatively sampled observations.
During Phase~1, the trajectories descend rapidly from $\theta_0 = 5$
toward the true value $\theta=3$.  Despite a large initial spread,
most runs have converged close to the true value by iteration~$100$.
The middle panel confirms this: the distribution of $\hat\theta$ across
repetitions after Phase~1 is concentrated near $\theta^* = 3$.
The final mean estimate across all repetitions is
$\hat\theta \approx 2.844$ (orange dotted line), after Phase~2.

\begin{figure}[htbp]
  \centering
  \includegraphics[width=\textwidth]{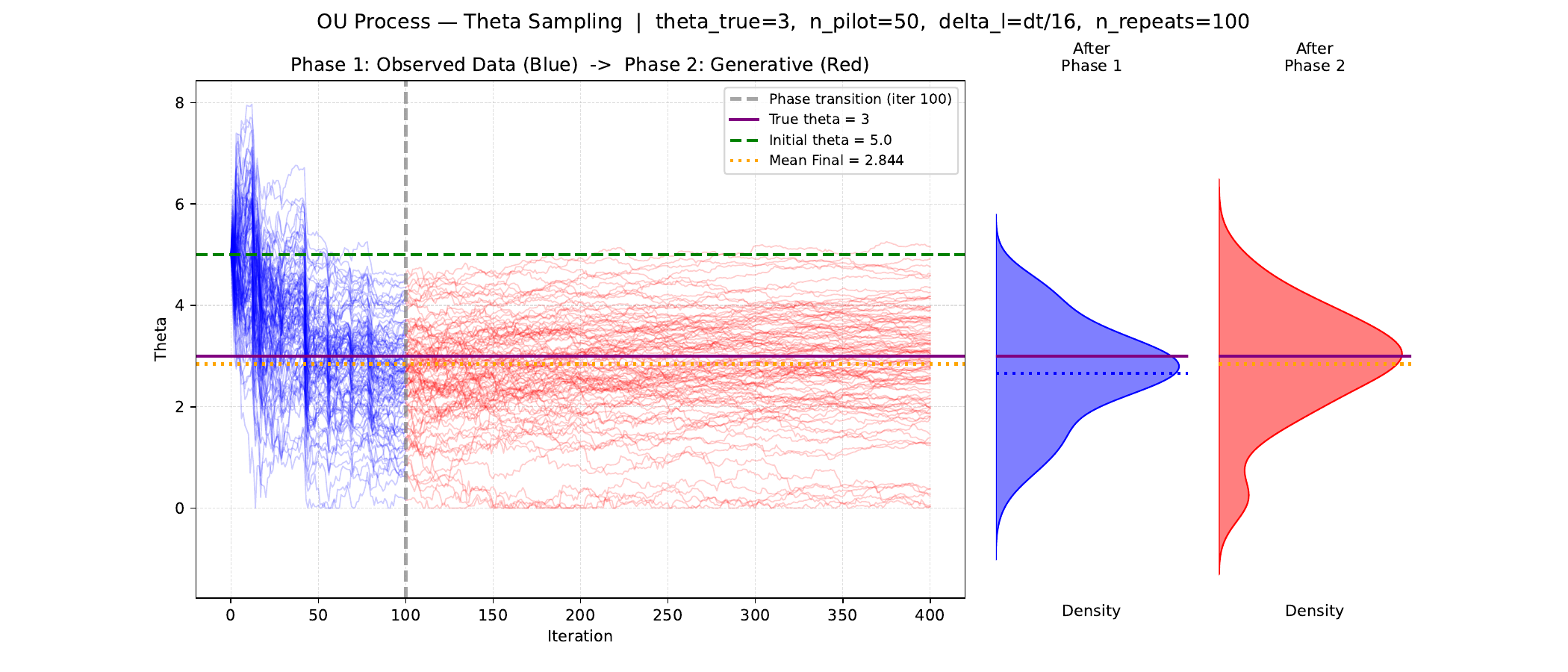}
  \caption{Parameter estimation trajectories over $100$ independent
    repetitions.  \textbf{Left}: blue curves show Phase~1 (observed data,
    iterations $0$--$100$); red curves show Phase~2 (generative
    continuation, iterations $100$--$400$).  The purple line marks the true
    value $\theta^* = 3$, the green dashed line the initial value
    $\theta_0 = 5$, and the orange dotted line the final mean
    $\hat\theta \approx 2.844$.  \textbf{Middle and right}: distribution of $\hat\theta$ across
    repetitions at the end of Phase~1 and Phase~2, respectively.}
  \label{fig:ou-result}
\end{figure}

\subsection{Example 2}

The stochastic Lotka--Volterra (SLV) system for prey $X_t^1$ and predator $X_t^2$
is
\begin{align*}
  \mathrm{d}X_t^1 &= X_t^1\!\left(\alpha - \beta X_t^2\right)\mathrm{d}t
                    + \sigma_1 X_t^1\,\mathrm{d}W_t^1, \\
  \mathrm{d}X_t^2 &= X_t^2\!\left(\zeta X_t^1 - \gamma\right)\mathrm{d}t
                    + \sigma_2 X_t^2\,\mathrm{d}W_t^2,
\end{align*}
with independent Brownian motions $W^1, W^2$ and unknown parameter
$\theta = (\alpha, \beta, \zeta, \gamma)\in(\mathbb{R}^{+})^4$.  The
multiplicative noise and nonlinear cross-coupling make the transition
density intractable.  We will use the following auxiliary process. 
\begin{align*}
  \mathrm{d}\tilde X_t^1 &= \tilde X_t^1
    \left[(\alpha - \beta x^2)\!\left(1 - \tfrac{t}{t_1}\right)
         + (\alpha - \beta {x'}^2)\tfrac{t}{t_1}\right]\mathrm{d}t
    + \sigma_1 \tilde X_t^1\,\mathrm{d}W_t^1, \\
  \mathrm{d}\tilde X_t^2 &= \tilde X_t^2
    \left[(\zeta x^1 - \gamma)\!\left(1 - \tfrac{t}{t_1}\right)
         + (\zeta {x'}^1 - \gamma)\tfrac{t}{t_1}\right]\mathrm{d}t
    + \sigma_2 \tilde X_t^2\,\mathrm{d}W_t^2.
\end{align*}

We generate 100 data points with time spacing 0.1 with true parameters $\alpha^* = 1.0$, $\beta^* = 0.5$,
$\zeta^* = 0.3$, $\gamma^* = 0.8$, with diffusion coefficients
$\sigma_1 = 0.2$ and $\sigma_2 = 0.15$.  All four drift parameters are
treated as unknown and estimated jointly; $\sigma_1$ and $\sigma_2$ are
assumed known.  
Each observation interval is further subdivided into $8$
Euler--Maruyama sub-steps of size $\Delta_l = 0.0125$.  The algorithm is
initialized at $\theta_0 = (0.5,\, 1.0,\, 0.5,\, 0.5)$.  The step size
$ \gamma_k = 0.5/(k + 100)$.
After the Phase~1 process,  Phase~2 extends the trajectory by $100$
generative steps with the iteration counter continuing seamlessly.  The
entire procedure is repeated $50$ times with
independent random seeds, all sharing the same observed dataset.


Figure~\ref{fig:lv-result} shows the estimation trajectories for all four
parameters across $50$ independent repetitions.  Each row corresponds to
one parameter; the blue curves cover Phase~1 (iterations $0$--$100$) and
the red curves cover Phase~2 (iterations $100$--$200$).
For $\alpha$ and $\beta$, Phase~1 achieves rapid and reliable convergence.
The trajectories for $\alpha$ rise from the initial value $0.5$ toward the
true value $1.0$, and by iteration~$100$ the estimates are tightly
concentrated near $1.0$.  Similarly, $\beta$ descends from $1.0$ toward
$0.5$ and converges well within Phase~1.  The final mean estimates,
$\hat\alpha \approx 0.971$ and $\hat\beta \approx 0.498$, are close to
their respective true values with small bias.

Convergence is slower for $\zeta$ and $\gamma$.  The $\zeta$ trajectories
descend from $0.5$ toward $0.3$ but with noticeably more spread across
runs, and some trajectories dip below zero during Phase~1.  The final
mean $\hat\zeta \approx 0.333$ lies slightly above the true value $0.3$.
For $\gamma$, the estimates rise from $0.5$ but do not fully reach the
true value $0.8$ by the end of Phase~1.
In Phase~2, the distributions of all four parameters widen relative to
the end of Phase~1, consistent with the behaviour observed in the OU
experiment: once real observations are no longer available, the generative
continuation introduces additional variance without the anchoring effect
of data.

\begin{figure}[htbp]
  \centering
  \includegraphics[width=\textwidth,height=12cm]{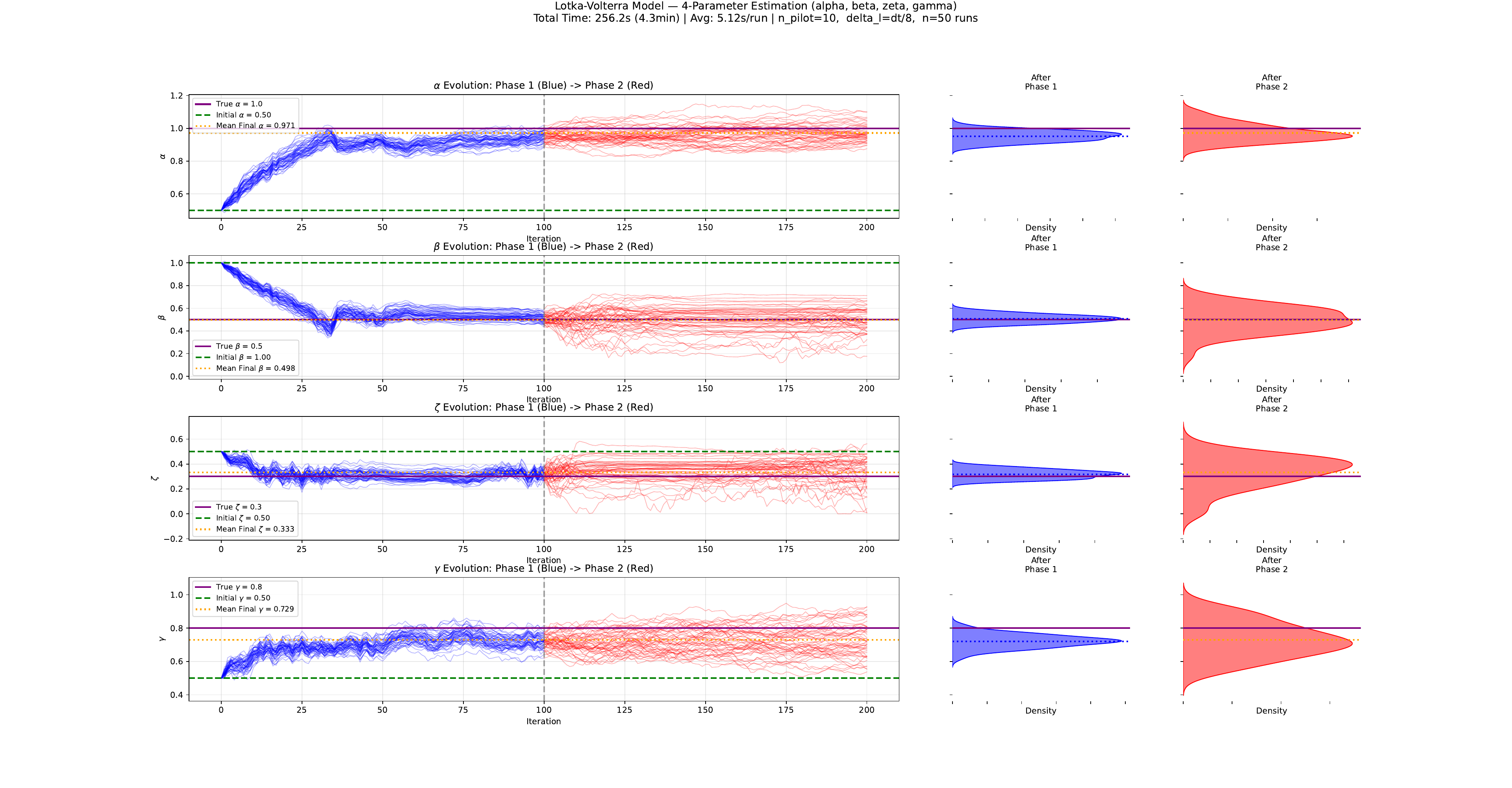}
  \caption{Four-parameter estimation trajectories over $50$ independent
    repetitions.  Each row shows one parameter ($\alpha$, $\beta$, $\zeta$,
    $\gamma$); blue curves are Phase~1 (observed data, iterations
    $0$--$100$) and red curves are Phase~2 (generative continuation,
    iterations $100$--$200$).  Purple lines mark true values; green dashed
    lines mark initial values; orange dotted lines mark final mean
    estimates.  Side panels show the distribution of estimates across
    repetitions at the end of each phase.}
  \label{fig:lv-result}
\end{figure}

\subsubsection{Comparison with MCMC}

We compare our approach with the state-of-the-art MCMC method in \cite{non_synch}.
This approach is based upon multilevel Monte Carlo \cite{giles} and also uses the bridge method
adopted in this paper.  Although it is designed for non-snychronous observations it is easily adapted to
the problem in this paper.   We compare the scenario of Figure \ref{fig:lv-result} to using MCMC in this section.
The MCMC is the multilevel approach in \cite{non_synch} with the priors used in that paper for the unknown parameters we consider.

Figure~\ref{fig:compare-bd-sgd} compares the empirical marginal
distributions produced by the two methods on the same observed
trajectory.  The martingale-score density (blue) is obtained as a
kernel density estimate over $100$ independent
repetitions; the MCMC density (red) is a kernel density estimate
over $119$ thinned draws extracted from a chain of length $22000$
at the finest discretization level.  We note that the comparison here is mainly for the purposes of
computational time/efficiency.  The Bayesian approach in \cite{non_synch} samples from an approximate version of the posterior of the parameters and corrects via importance sampling,  so intrinsically the samples are displayed just to provide idea of the spread,  rather than for inference per-se.

The total wall-clock time is $25.5$ minutes for the
martingale-score method and $181.4$ minutes for MCMC.  The
per-co-ordinate effective sample sizes (computed via the method of \cite{gey1,gey2}) of the MCMC chain are $135.7$,
$132.9$, $122.3$, and $118.9$ for $\alpha$, $\beta$, $\zeta$, and
$\gamma$ respectively, all of which exceed $100$; the worst-case
integrated autocorrelation time, $185.1$ on $\gamma$, sets the
thinning step.  Translated to cost per independent sample,
 is $15.3$ seconds for the martingale-score method
and $91.5$ seconds for MCMC, a factor of approximately $6$.  

\begin{figure}[htbp]
  \centering
  \includegraphics[width=\textwidth]{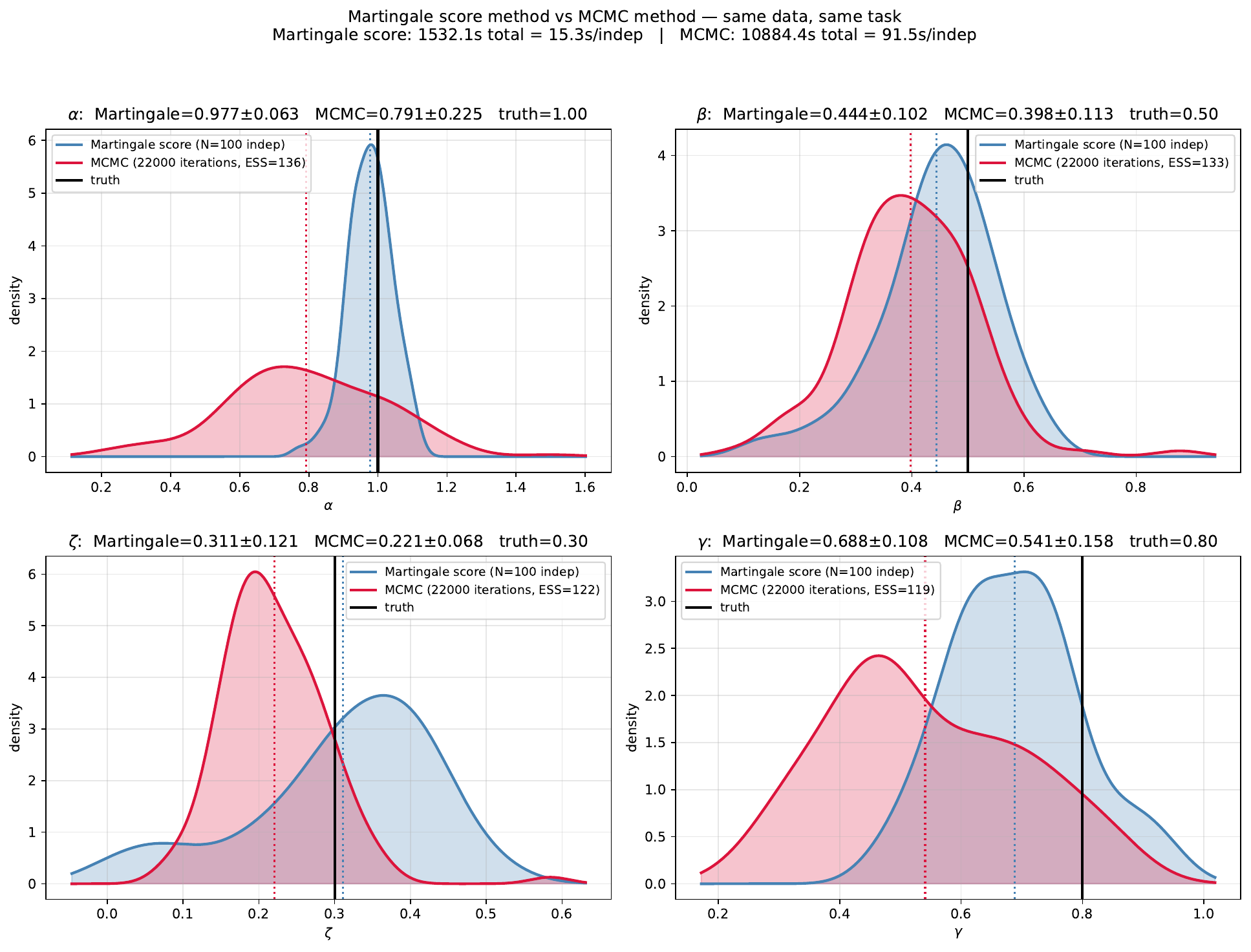}
  \caption{Marginal densities of the four drift parameters
    $(\alpha, \beta, \zeta, \gamma)$ on the same observed SLV trajectory
    used in Figure \ref{fig:lv-result}, comparing the
    martingale-score gradient method (blue, kernel density over
    $100$ independent repetitions) with the MLPMMH
    sampler of \cite{non_synch} (red, kernel density over
    $119$ approximately independent draws thinned from a chain of
    $22000$ iterations at discretization level 8).  Solid black lines mark the
    true parameter values; dotted vertical lines mark the empirical
    posterior means.  The header reports total wall time and cost per
    independent sample for each method.}
  \label{fig:compare-bd-sgd}
\end{figure}

\section{Summary and Future Work}\label{sec:sum}

In this article we have developed a new score-based martingale method for diffusions that are observed with low frequency.  We have proved,  under assumptions,  that the estimator is close to a continuous-time version of the algorithm uniformly in time.  We implemented our method in several examples and have established that it can produce statistically meaningful inferences at computational times which are far less than state-of-the-art MCMC algorithms.

There are several avenues for future research.  The first is to consider an extension to use the multilevel Monte Carlo method.  The idea would be to consider expectations of functions associated to a hierarchy of martingale posteriors with falling levels of time-discretization.  Then one can attempt to utilize the multilevel identity to reduce the computational effort as is typical in that literature.  A second important area is that of high-frequency observations. In such cases one could envisage time steps of the martingale posterior at the level of frequency of the observation, as was done here,  except simple approximations, such as Euler-Maruyama densities could work rather well,  leading to a fast method, albeit at high frequency.  Thirdly is to consider extensions to other models,  such as partially observed diffusions (e.g.~\cite{jasra}) and an investigation in that case may well be warranted as existing (exact) computational methods can be notoriously slow.

\appendix

\section{Theoretical Exposition}\label{app:theory}

\subsection{Perfect \textbf{MPD}}

The following is the \textbf{MPD} with $l\rightarrow\infty$; for brevity we omit the first phase,  but inspection of this second phase,  it should be clear what the first phase would be.
This is the reference method we will compare the iterates of \textbf{MPD} to.  In the below and throughout we suppose $\theta_0=\theta_0^l$ which is the initialization used for 
\textbf{MPD}.

\begin{enumerate}
\item{Input $\theta_T$ and $X_T$ from the first phase, set $k=T+1$.  Sample $X_k|x_{k-1}$ using 
$$
f_{\theta_{k-1},k-1,k}(x_k|x_{k-1}) \propto \int R_{\theta_{k-1},k-1,k}\left(
C_{\theta_{k-1},k-1,k}(x_{k-1},\mathbf{w}_{[k-1,k]},x_k)\right) \mathbb{W}(d\mathbf{w}_{[k-1,k]})
$$
 Go to 2..} 
\item{Iterate: 
\begin{itemize}
\item{Sample $\mathbf{W}_{[k-1,k]}|x_{k},x_{k-1}$ from
$$
\pi_{\theta_{k-1}}(d\mathbf{w}_{[k-1,k]}|x_{k},x_{k-1}) \propto
R_{\theta_{k-1},k-1,k}\left(
C_{\theta_{k-1},k-1,k}(x_{k-1},\mathbf{w}_{[k-1,k]},x_{k})\right) \mathbb{W}(d\mathbf{w}_{[k-1,k]})
$$
compute $\nabla \log\left(R_{\theta_{k-1},k-1,k}\left(
C_{\theta_{k-1},k-1,k}(x_{k-1},\mathbf{w}_{[k-1,k]},x_k)\right)\right)$.}
\item{Sample $\mathbf{V}_{[k-1,k]},u_k|x_{k-1}$
$$
\pi_{\theta_{k-1}}(d\mathbf{v}_{[k-1,k]},u_k|x_{k-1}) \propto
R_{\theta_{k-1},k-1,k}\left(
C_{\theta_{k-1},k-1,k}(x_{k-1},\mathbf{v}_{[k-1,k]},u_k)\right) \mathbb{W}(d\mathbf{v}_{[k-1,k]})
$$
and compute $\nabla \log\left(R_{\theta_{k-1},k-1,k}\left(
C_{\theta_{k-1},k-1,k}(x_{k-1},\mathbf{v}_{[k-1,k]},u_k)\right)\right)$.}
\item{Update:
\begin{align*}
\theta_k& = \theta_{k-1} + \gamma_k \Big\{
\nabla \log(R_{\theta_{k-1},k-1,k}\left(
C_{\theta_{k-1},k-1,k}(x_{k-1},\mathbf{w}_{[k-1,k]},x_k)\right)) - \\
& \nabla \log(R_{\theta_{k-1},k-1,k}\left(
C_{\theta_{k-1},k-1,k}(x_{k-1},\mathbf{v}_{[k-1,k]},u_k)\right))
\Big\}.
\end{align*}
}
\end{itemize}
Go to 3..
}
\item{Set $k=k+1$ and sample $X_k|x_{k-1}$ using 
$$
f_{\theta_{k},k-1,k}(x_k|x_{k-1}) \propto \int R_{\theta_k,k-1,k}\left(
C_{\theta_k,k-1,k}(x_{k-1},\mathbf{w}_{[k-1,k]},x_k)\right) \mathbb{W}(d\mathbf{w}_{[k-1,k]})
$$
Go to 2..}
\end{enumerate}

\subsection{Notations and Conventions}

The basic idea of the proof of Theorem \ref{theo:main} is based upon an interpolation type argument which is used in the analysis of the stability of the filter and numerical methods for ordinary differential equations; see for example \cite[Chapter 7]{delmoral} and the references therein.  To that end let $s\in\{\emptyset,l\}$ 
be a subscript 
where $l\in\mathbb{N}$ is given and fixed and if $s=\emptyset$ then there is no subscript (i.e.~it is null) and define the recursion for $m\geq 0$ fixed,  $k\geq m$
$$
\theta_{k+1}^s(m,\theta_m^s,X_m^s) = \theta_{k}^s(m,\theta_m^s,X_m^s) + \gamma_k \mathsf{S}_{k,k+1}^s\left(
X_{k:k+1}^s,U_{k+1}^s,\theta_{k}^s(m,\theta_m^s,X_m^s),\mathbf{w}_{[k,k+1]}^s,\mathbf{v}_{[k,k+1]}^s\right)
$$
with $\theta_{m}^s(m,\theta_m^s,X_m^s)=\theta_m^s$ and
where
$$
\mathsf{S}_{k,k+1}^s\left(
X_{k:k+1}^s,U_{k+1}^s,\theta_{k}^s(m,\theta_m^s,X_m^s),\mathbf{w}_{[k,k+1]}^s,\mathbf{v}_{[k,k+1]}^s\right)  =
$$
$$
\nabla \log(R_{\theta_{k}^s(m,\theta_m^s,X_m^s),k,k+1}^s\left(
C_{\theta_{k}^s(m,\theta_m^s,X_m^s),k,k+1}^s(x_{k}^s,\mathbf{w}_{[k,k+1]}^s,x_{k+1}^s)\right)) - 
$$
$$
\nabla \log(R_{\theta_{k}^s(m,\theta_m^s,X_m^s),k,k+1}^s\left(
C_{\theta_{k}^s(m,\theta_m^s,X_m^s),k,k+1}^s(x_{k}^s,\mathbf{v}_{[k,k+1]}^s,u_{k+1}^s)\right))
$$
and we use the $s$ subscripts for $X$ and $U$ to denote whether they have been generated by \textbf{MPD}
or the perfect version;  note that in the first phase,  of course,  the $X$ are the same in both \textbf{MPD} and the perfect one.  We note that $\theta_n(0,\theta_0,x_0)$ represents the perfect \textbf{MPD} recursion and
 $\theta_n^l(0,\theta_0,x_0)$ the \textbf{MPD} one can actually implement in practice (i.e.~up-to time discretization). 
Our proof will focus on the decomposition
\begin{equation}\label{eq:telescope_sum}
\theta_n(0,\theta_0,x_0) - \theta_n^l(0,\theta_0,x_0)  = \sum_{j=0}^{n-1}
\left\{
\theta_n\left(j,\theta_{j}^l(0,\theta_0,X_0),X_j^l\right) -
\theta_n\left(j+1,\theta_{j+1}^l(0,\theta_0,X_0),X_{j+1}^l)
\right)
\right\}.
\end{equation}

\subsection{Assumptions}\label{app:theory_ass}

Our objective is leverage the identity that is given in \eqref{eq:telescope_sum}.  To that end we can construct a process that is equal in distribution to the random variables that are given in \eqref{eq:telescope_sum}. The reason why this will be useful will be clarified once the process and assumptions are detailed.
We consider each summand of \eqref{eq:telescope_sum},  denote it $\mathcal{S}_j$,  individually and suppose that one can simulate each difference independently across the summands.  On inspection of the summand,  $\theta_n\left(j,\theta_{j}^l(0,\theta_0,X_0),X_j^l\right)$ 
(resp.~$\theta_n\left(j+1,\theta_{j+1}^l(0,\theta_0,X_0),X_{j+1}^l)\right)$
consists of running the \textbf{MPD} up-to iterative time $j$ (resp.~time $j+1$) and then generating the perfect \textbf{MPD} starting at $\theta_{j}^l(0,\theta_0,X_0)$ and $X_j^l$ at time $j$
(resp.~$\theta_{j+1}^l(0,\theta_0,X_0)$ and $X_{j+1}^l$ at time $j+1$)
 up-to time $n$ (resp.~$n$).   We suppose that we can construct a coupling of the laws associated to the sequences of random variables as just described,  with expectation operator $\check{\mathbb{E}}$,  such that the to be described assumptions are satisfied.  We discuss the extent to which these assumptions have previously been investigated,  below.  To minimize any notational issues,  we write, for example $X_j$, $X_j^l$ without reference to the index $j$ in $\mathcal{S}_j$,  but it should be understood that the assumption holds for every (independent) process $\mathcal{S}_j$.   The constants below are assumed $j$ independent.  To shorten the notation set
\begin{align*}
\mathcal{D}_j^l := &\mathsf{S}_{j,j+1}\left(X_j^l,X_{j+1},U_{j},\theta_j^l(0,\theta_0,X_0),\mathbf{W}_{[j,j+1]},\mathbf{V}_{j,j+1]}\right) - \\ & \mathsf{S}_{j,j+1}^l\left(X_{j:j+1}^l,U_{j}^l,
\theta_{j}^l(0,\theta_0,X_0),\mathbf{W}_{[j,j+1]}^l,\mathbf{V}_{j,j+1]}^l\right).
\end{align*}

\begin{hypA}\label{ass:1}
There exists a $C<+\infty$ such that for any $l\in\mathbb{N}$
$$
\max\left\{\check{\mathbb{E}}\left[\left\|\mathcal{D}_j^l\right\|^2\right],
\left|\check{\mathbb{E}}\left[\mathcal{D}_j^l\right]\right|,
\check{\mathbb{E}}\Big[\Big\|X_{j+1}-X_{j+1}^l\Big\|^2\Big]
\right\}
\leq C\Delta_l.
$$
\end{hypA}

\begin{hypA}\label{ass:2}
There exists a non-negative sequence $(\alpha_n)_{n\geq 0}$,  such that $\sup_{n\geq 1,j\geq 1, j<n}\sum_{i=j}^n \alpha_{n-i}<+\infty$ and $C<+\infty$ such that for any $(n,j,x,x',\theta,\theta')\in\mathbb{N}\times\{1,\dots,n-1\}\times\mathbb{R}^{2d}\times\Theta^2$
\begin{eqnarray*}
\check{\mathbb{E}}\left[\|\theta_n(j+1,\theta,x)-\theta_n(j+1,\theta',x')\|^2\right] & \leq & C\alpha_{n-j}\left\{
\|x-x'\|^2+\|\theta-\theta'\|^2
\right\} \\
|\check{\mathbb{E}}\left[\theta_n(j+1,\theta,x)-\theta_n(j+1,\theta',x')\right]| & \leq & C\alpha_{n-j}\left\{
|x-x'|+|\theta-\theta'|
\right\}.
\end{eqnarray*}
\end{hypA}

The assumptions deserve a discussion to understand their implication and whether they can be verified. 
(A\ref{ass:1}) concerns the accuracy of the time-discretization, interplayed with the updates of the parameter up-to time $j+1$.  In essence,  one considers the time-discretized updates of the parameter at a given $l$ up-to time $j$,  which could be identical for both  $\theta_n\left(j,\theta_{j}^l(0,\theta_0,X_0),X_j^l\right)$ and $\theta_n\left(j+1,\theta_{j+1}^l(0,\theta_0,X_0),X_{j+1}^l)
\right)$ (this would be a valid coupling) and then in one step $\theta_n\left(j,\theta_{j}^l(0,\theta_0,X_0),X_j^l\right)$ transitions to the perfect \textbf{MPD} and the other stays in the discretized one.  The assumption states that in this transition,  one can couple the simulation so that the standard Euler-Maruyama rate e.g.~\cite{kloeden} (which holds for sufficiently regular diffusions as assumed in this paper) is maintained.  This seems intutively possible and has been studied numerically in the context of diffusion bridges in \cite{bridges,non_synch} who develop MCMC methods to achieve these rates.  In those articles the technical barrier to proving these rates are discussed in details and we refer the reader to those papers.  (A\ref{ass:2}) is a type of forgetting property of the update of the parameters and the intermediate simulation,  when considering the perfect \textbf{MPD}.  Essentially,  it says that for long-time the expected square $L_2-$distance between two perfect \textbf{MPD} will disappear.  Normally these type of properties can be verified for Markov chains (of which our process is) with sufficiently fast mixing; see \cite{delmoral} for example.  Ultimately one would like to prove (A\ref{ass:1}-\ref{ass:2}) but we expect that this a particularly arduous challenge which is well out of the scope of the current article,  but despite this,  we believe that these assumptions can hold under appropriate conditions. 

\subsection{Proof of Theorem \ref{theo:main}}\label{app:theory_prf}

\begin{proof}
Throughout  $C$ is a constant that depends on $d_{\theta}$ only and whose value may change from line-to-line.
We begin by considering \eqref{eq:telescope_sum},  of which one has the bound:
\begin{equation}\label{eq:mainprf1}
\check{\mathbb{E}}\left[\|\sum_{j=0}^{n-1}\mathcal{S}_j\|^2\right] \leq 
\sum_{j=0}^{n-1}\check{\mathbb{E}}\left[\|\mathcal{S}_j\|^2\right] +
C\sum_{j<k}|\check{\mathbb{E}}\left[\mathcal{S}_j\right]||\check{\mathbb{E}}\left[\mathcal{S}_k\right]|
\end{equation}
where we have used the independence of $\mathcal{S}_j$ and $\mathcal{S}_k$, and we have used basic properties of $L_1-$norms.  We consider a summand of \eqref{eq:telescope_sum} which is equal to
\begin{equation}\label{eq:mainprf2}
\mathcal{S}_j = \sum_{i=1}^3 \mathcal{S}_j(i)
\end{equation}
where
\begin{eqnarray*}
\mathcal{S}_j(1) & = &
\theta_0 + \sum_{i=1}^{j} \gamma_i 
\mathsf{S}_{i-1,i}^l\left(
X_{i-1:i}^l,U_{i}^l,\theta_{i-1}^l(0,\theta_0,X_0),\mathbf{w}_{[i-1,i]}^l,\mathbf{v}_{[i-1,i]}^l\right) - \\
& &  
\left\{\theta_0 + \sum_{i=1}^{j} \gamma_i 
\mathsf{S}_{i-1,i}^l\left(
X_{i-1:i}^l,U_{i}^l,\theta_{i-1}^l(0,\theta_0,X_0),\mathbf{w}_{[i-1,i]}^l,\mathbf{v}_{[i-1,i]}^l\right) \right\}
\\
\mathcal{S}_j(2) & = &
\gamma_{j+1} 
\mathsf{S}_{j,j+1}\left(X_j^l,X_{j+1},U_{j},\theta_j^l(0,\theta_0,X_0),\mathbf{w}_{[j,j+1]},\mathbf{v}_{j,j+1]}\right) - \\
& & 
\gamma_{j+1} 
\mathsf{S}_{j,j+1}^l\left(X_{j:j+1}^l,U_{j}^l,
\theta_{j}^l(0,\theta_0,X_0),\mathbf{w}_{[j,j+1]}^l,\mathbf{v}_{j,j+1]}^l\right)
\\
\mathcal{S}_j(3) & = &
\sum_{i=j+2}^{n} \gamma_i 
\mathsf{S}_{i-1,i}\left(
X_{i-1:i},U_{i},\theta_{i-1}(j,\theta_j^l(0,\theta_0,X_0),X_j^l),\mathbf{w}_{[i-1,i]},\mathbf{v}_{[i-1,i]}\right) - \\
&  &\Big\{
\gamma_{j+2} 
\mathsf{S}_{j+1,j+2}\left(X_{j+1}^l,X_{j+2},U_{j+2},
\theta_{j+1}^l(0,\theta_0,X_0),\mathbf{w}_{[j+1,j+2]},\mathbf{v}_{j+1,j+2]}\right) \\
& &
\sum_{i=j+3}^{n} \gamma_i 
\mathsf{S}_{i-1,i}\left(
X_{i-1:i},U_{i},\theta_{i-1}(j+1,\theta_{j+1}^l(0,\theta_0,X_0),X_{j+1}^l),\mathbf{w}_{[i-1,i]},\mathbf{v}_{[i-1,i]}\right)
\Big\}
\end{eqnarray*}
where we use the convention that the sum over an empty set is zero 
and note that $\mathcal{S}_j(1)=0$;  it is just included for clarity of exposition.   As a result, to bound
\eqref{eq:mainprf1} we will use the decomposition that is given in \eqref{eq:mainprf2} which consistutes dealing with
terms of the type $\check{\mathbb{E}}\left[\|\mathcal{S}_j\|^2\right]$ and $|\check{\mathbb{E}}\left[\mathcal{S}_j\right]|$;  we begin with the former.

We have that 
$$
\check{\mathbb{E}}\left[\|\mathcal{S}_j\|^2\right] \leq C
\left(
\check{\mathbb{E}}\left[\|\mathcal{S}_j(2)\|^2\right] + \check{\mathbb{E}}\left[\|\mathcal{S}_j(3)\|^2\right]
\right).
$$
Then applying (A\ref{ass:1}) to the first term on the R.H.S.~and (A\ref{ass:2}) to the second term on the R.H.S.~we have that
$$
\check{\mathbb{E}}\left[\|\mathcal{S}_j\|^2\right] \leq C
\left(\gamma_{j+1}^2\Delta_l + \alpha_{n-j}\left\{
\check{\mathbb{E}}\left[\|X_{j+1}-X_{j+1}^l\|^2\right]
+\check{\mathbb{E}}\left[\|
\theta_{j+1}(j,\theta_j^l(0,\theta_0,X_0),X_{j+1}^l)-
\theta_{j+1}^l(0,\theta_0,X_0)\|^2\right]\right\}
\right).
$$
Then applying (A\ref{ass:1}) to $\check{\mathbb{E}}\left[\|X_{j+1}-X_{j+1}^l\|^2\right]$ we obtain
$$
\check{\mathbb{E}}\left[\|\mathcal{S}_j\|^2\right] \leq C
\left(\gamma_{j+1}^2\Delta_l + \alpha_{n-j}\left\{\Delta_l + \check{\mathbb{E}}\left[\|
\theta_{j+1}(j,\theta_j^l(0,\theta_0,X_0),X_{j+1}^l)-
\theta_{j+1}^l(0,\theta_0,X_0)\|^2\right]\right\}
\right)
$$
so to upper-bound $\check{\mathbb{E}}\left[\|\mathcal{S}_j\|^2\right]$ we now consider
$\check{\mathbb{E}}\left[\|
\theta_{j+1}(j,\theta_j^l(0,\theta_0,X_0),X_{j+1}^l)-
\theta_{j+1}^l(0,\theta_0,X_0)\|^2\right]$.  Using the iterative update formula of the \textbf{MPD} we
have
$$
\theta_{j+1}(j,\theta_j^l(0,\theta_0,X_0),X_{j+1}^l)-
\theta_{j+1}^l(0,\theta_0,X_0) = \gamma_{j+1}\mathcal{D}_j^l
$$
so that by using (A\ref{ass:1}) we have that 
\begin{equation}\label{eq:mainprf3}
\check{\mathbb{E}}\left[\|
\theta_{j+1}(j,\theta_j^l(0,\theta_0,X_0),X_{j+1}^l)-
\theta_{j+1}^l(0,\theta_0,X_0)\|^2\right] \leq 
C\gamma_{j+1}^2\Delta_l.
\end{equation}
Thus we have shown that
\begin{equation}\label{eq:mainprf4}
\check{\mathbb{E}}\left[\|\mathcal{S}_j\|^2\right] \leq C\Delta_l
\left(\gamma_{j+1}^2 + \alpha_{n-j}\left\{1 + 
\gamma_{j+1}^2\right\}\right).
\end{equation}

For $|\check{\mathbb{E}}\left[\mathcal{S}_j\right]|$,  we note that if $j\geq T$ this term is exactly zero by the martingale property,  so we suppose that $j<T$.  We have that 
$$
|\check{\mathbb{E}}\left[\mathcal{S}_j\right]|
 \leq C
\left(
|\check{\mathbb{E}}\left[\mathcal{S}_j(2)\right]| + |\check{\mathbb{E}}\left[\mathcal{S}_j(3)\right]|
\right).
$$
Then applying (A\ref{ass:1}) to the first term on the R.H.S.~and (A\ref{ass:2}) to the second term on the R.H.S.~we have that
$$
|\check{\mathbb{E}}\left[\mathcal{S}_j\right]|
 \leq C
\left(
\gamma_{j+1}\Delta_l + \alpha_{n-j}\left\{
\check{\mathbb{E}}\left[|X_{j+1}-X_{j+1}^l|\right]
+\check{\mathbb{E}}\left[|
\theta_{j+1}(j,\theta_j^l(0,\theta_0,X_0),X_{j+1}^l)-
\theta_{j+1}^l(0,\theta_0,X_0)|\right]\right\}
\right).
$$
Then as 
\begin{eqnarray*}
\check{\mathbb{E}}\left[|X_{j+1}-X_{j+1}^l|\right] & \leq & C\check{\mathbb{E}}\left[\|X_{j+1}-X_{j+1}^l\|^2\right]^{1/2}\\
& \leq & C\Delta_l^{1/2}
\end{eqnarray*}
where we have used (A\ref{ass:1}
and similarly
\begin{eqnarray*}
\check{\mathbb{E}}\left[|
\theta_{j+1}(j,\theta_j^l(0,\theta_0,X_0),X_{j+1}^l)-
\theta_{j+1}^l(0,\theta_0,X_0)|\right] & \leq & C
\check{\mathbb{E}}\left[\|
\theta_{j+1}(j,\theta_j^l(0,\theta_0,X_0),X_{j+1}^l)-
\theta_{j+1}^l(0,\theta_0,X_0)\|^{2}\right]^{1/2} \\
& \leq & C\gamma_{j+1}\Delta_l^{1/2}
\end{eqnarray*}
where we have used \eqref{eq:mainprf3},  one obtains
\begin{eqnarray}
|\check{\mathbb{E}}\left[\mathcal{S}_j\right]|
 & \leq & C
\left(\gamma_{j+1}\Delta_l + \alpha_{n-j}\left\{
\Delta_l^{1/2} + \gamma_{j+1}\Delta_l^{1/2}
\right\}
\right) \nonumber \\
 & \leq & C\Delta_l^{1/2}
\left(\gamma_{j+1} + \alpha_{n-j}\left\{
1 + \gamma_{j+1}
\right\}
\right).\label{eq:mainprf5}
\end{eqnarray}

Now,  combining \eqref{eq:mainprf1} with \eqref{eq:mainprf4} and \eqref{eq:mainprf5} one has that
$$
\check{\mathbb{E}}\left[\|\sum_{j=0}^{n-1}\mathcal{S}_j\|^2\right] \leq  
C\left(\sum_{j=0}^{n-1}\Delta_l
\left(\gamma_{j+1}^2 + \alpha_{n-j}\left\{1 + 
\gamma_{j+1}^2\right\}\right) +
\left\{\sum_{j=0}^{T-1}
\Delta_l^{1/2}
\left(\gamma_{j+1} + \alpha_{n-j}\left\{
1 + \gamma_{j+1}
\right\}
\right)
\right\}^2\right).
$$
As all the sums are upper-bounded uniformly in $n$ we have obtained that 
$$
\sup_{n\geq 1}\check{\mathbb{E}}\left[\left\|\theta_n^l-\theta_n^l\right\|^2\right] \leq C\Delta_l.
$$
as was to be proved.
\end{proof}

\end{document}